\documentclass[runningheads, ]{llncs}
\usepackage[T1]{fontenc}
\usepackage{graphicx}
\usepackage{hyperref}
\usepackage{amsfonts, amssymb, amsmath}
\usepackage[ruled,vlined,linesnumbered]{algorithm2e}
\usepackage{subfigure}
\begin{document}
\title{Computational Power of Mobile Robots: Discrete Version}
%
%
\author{Avisek Sharma\orcidID{0000-0001-8940-392X} \and
Pritam Goswami\orcidID{0000-0002-0546-3894} \and
Buddhadeb Sau\orcidID{0000-0001-7008-6135}}

\authorrunning{A Sharma, P Goswami, and B Sau}
\institute{Department of Mathematics, Jadavpur University, India\\ \email{aviseks.math.rs@jadavpuruniversity.in, pgoswami.academic@gmail.com, buddhadeb.sau@jadavpuruniversity.in}}

%
\maketitle              
\begin{abstract}
In distributed computing by mobile robots, robots are deployed over a region, continuous or discrete, operating through a sequence of \textit{look-compute-move} cycles. An extensive study has been carried out to understand the computational powers of different robot models. The models vary on the ability to 1)~remember constant size information and 2)~communicate constant size message. Depending on the abilities the different models are 1)~$\mathcal{OBLOT}$ (robots are oblivious and silent), 2)~$\mathcal{FSTA}$ (robots have finite states but silent), 3)~$\mathcal{FCOM}$ (robots are oblivious but can communicate constant size information) and, 4)~$\mathcal{LUMI}$ (robots have finite states and can communicate constant size information). Another factor that affects computational ability is the scheduler that decides the activation time of the robots. The main three schedulers are \textit{fully-synchronous}, \textit{semi-synchronous} and \textit{asynchronous}. Combining the models ($M$) with schedulers ($K$), we have twelve combinations $M^K$. 

In the euclidean domain, the comparisons between these twelve variants have been done in different works for transparent robots, opaque robots, and robots with limited visibility. There is a vacant space for similar works when robots are operating on discrete regions like networks. It demands separate research attention because there have been a series of works where robots operate on different networks, and there is a fundamental difference when robots are operating on a continuous domain versus a discrete domain in terms of robots' movement. This work contributes to filling the space by giving a full comparison table for all models with two synchronous schedulers: fully-synchronous and semi-synchronous.

\end{abstract}

\section{Introduction}

In distributed computing swarm robotics is a well-studied area. The last two decades have been dedicated to studying different problems like \textit{gathering}, \textit{pattern formation}, \textit{exploration}, and \textit{dispersion} etc. The majority of problems have been dealt with in both, first in continuous domains, then in discrete domains. Continuous domains are Euclidean planes, continuous cycles, etc. and discrete domains are different types of finite and infinite graphs. In these problems \textit{robots} are point computing units deployed over some domain, continuous or discrete. The robots are generally \textit{autonomous} (the robots do not have central control), \textit{identical} (indistinguishable from physical appearance), \textit{anonymous} (the robots do not have any unique identifier) and \textit{homogeneous} (all robots have same capabilities and runs the same algorithm). The robots operate in \textit{look-compute-move} (LCM) cycles. On activation, a robots enter in look phase where it finds out the positions of the other robots in its vicinity. Then it enters into compute phase where it runs the inbuilt algorithm to find out where to move or stay put. Next, in move phase the robot moves or stays put according to the decisions in compute phase. In continuous domains, the robots have the freedom to move at any distance with any precision. On the other hand, in discrete domain or in networks the robots are only allowed to move to an adjacent node of its current position. A robot cannot stop in between an edge connecting two nodes. Problem have been considered in discrete domains because sometimes in the practical field the infinite precision in movement by robots is not possible. Also, sometimes the domains are engraved with a predefined networks, like star graphs, line graphs, rectangular and triangular grids, etc. 

The robot models vary on robot capabilities. The two fundamental capabilities are memory and communication. There are main four models: 1)~$\mathcal{OBLOT}$, 2)~$\mathcal{FSTA}$, 3)~$\mathcal{FCOM}$ and 4)~$\mathcal{LUMI}$. The $\mathcal{OBLOT}$ model (formally introduced in \cite{SuzukiY96}) is the weakest one. In this model, the robots are oblivious, i.e., they do not have persistent memory to remember past actions or past configurations. Also, in this model robots are silent, i.e., they do not have any communication ability. In the $\mathcal{LUMI}$ model the robots are equipped with persistent lights that can take finitely many colors. This model is formally introduced in \cite{FPSY16}. At the end of the compute phase, a robot changes the color of the lights as determined by the algorithm. Each robot can see its own color in a look phase. Thus, it serves as a finite memory. Next, other robots can see its lights. Thus, it serves as a communication architecture through which it can communicate a finite bit of information to other robots. The $\mathcal{LUMI}$ model is the strongest in terms of memory and communication. The other two intermediate models are introduced in \cite{FSVY160}. We have the first intermediate model $\mathcal{FSTA}$. In the $\mathcal{FSTA}$ model lights of a robot are internal, i.e., it is only visible to itself. Thus, the robots only have finite persistent memory, but the robots are deprived of communication ability. The next intermediate model is $\mathcal{FCOM}$. In this model, the lights of a robot are external, i.e., it cannot see color of its light. Thus, it only gives the ability to communicate a finite number of bits but has no persistent memory. 


Next, another key aspect for solving a problem is the considered scheduler. The scheduler is considered an entity that decides when to activate a robot. There are two main types of scheduler. In \textit{synchronous} scheduler, the time is divided into equal rounds. In each round, a set of robots are activated and they simultaneously perform their LCM cycle. In the \textit{fully-synchronous} scheduler (\textsc{Fsync}), in each round, all the robots are activated by the scheduler. On the other hand, in \textit{semi-synchronous} scheduler (\textsc{Ssync}), introduced in \cite{SuzukiY96}, a subset of all robots are activated in a particular round. A \textit{fair} scheduler activates each robot infinitely often. Here we assume that the scheduler is fair. In an \textit{asynchronous} (\textsc{Async})scheduler, introduced in \cite{FPSY99}, the robots operate independently of each other and they have no synchrony at all. Also, the length of an LCM cycle of a robot can be unbounded and be different in different activations. A \textit{variant} $M^K$ is a model $M$ together with a scheduler $K$. For convenience we shall denote $M^{\textsc{Fsync}}$, $M^{\textsc{Ssync}}$, and $M^{\textsc{Async}}$ respectively by $M^F$, $M^S$, and $M^A$.

Let $\mathcal{X}$ and $\mathcal{Y}$ be two variants. Then $\mathcal{X}>\mathcal{Y}$ denotes that the  $\mathcal{X}$ is strictly powerful variants than $\mathcal{Y}$, $\mathcal{X}\equiv\mathcal{Y}$ denotes that the  $\mathcal{X}$ is computationally equivalent to $\mathcal{Y}$, and $\mathcal{X}\perp\mathcal{Y}$ denotes that $\mathcal{X}$ and $\mathcal{Y}$ are computationally equivalent. The works in \cite{BFKPSW21,BFKPSW22,FPSY16,FSS023,FSW19,KKNPS21} refine the computational landscape when transparent robots are operating on the euclidean plane. Then in \cite{FMMP24}, the authors refine the computational landscape when the opaque (non-transparent) robots are operating on the euclidean plane. In the above works, the robots have full visibility. In a similar work in \cite{DGSGS24}, the authors considered limited visibility and in this work also, the robots operate on euclidean plane. Similar works are required when the robots are operating on different types of networks. A vast research work has been dedicated to investigating problems like gathering, pattern formation, etc. on different networks. For example, in \cite{PP24} the rendezvous problem (gathering of two robots) has been dealt with in an arbitrary network, in \cite{BKLT24} the gathering problem has been considered in discrete ring or cycle graph, in \cite{SGGS24} the arbitrary pattern formation problem has been considered in an infinite rectangular grid. From the above-mentioned model, comparison works for continuous versions, we observe that almost every problem considered for comparison cannot be converted to get a similar result for discrete versions. So a fresh set of problems and a completely different research focus are required for the discrete version. 

In the opinion of the authors, comparison in the discrete domain is slightly more difficult than the comparison in the continuous domain. The reason is as follows. In the continuous domain, the robots can be directed to move at any distance and in any direction. But through finite memory or finite communication, it is not possible to remember or communicate a real number. This limitation has been exploited in the previous works. But in a discrete domain, moving distance is bounded by one hop, and possible directions is also finitely many, as a robot can move to any one of the finitely many adjacent nodes. So, the same method of exploitation does not work in discrete domains. So some completely different problems in nature need to be considered.

To the best of our knowledge, \cite{DSFN18} is the only work where authors refine the landscape when the robots are operating on graphs. In \cite{DSFN18}, authors did not consider the models $\mathcal{FSTA}$ and $\mathcal{FCOM}$. In this work, we give the full computational landscape for all variants $M^K$, where $M\in\{\mathcal{OBLOT},\ \mathcal{FSTA},\ \mathcal{FCOM},\ \mathcal{LUMI}\}$ and $K\in\{\textsc{Fsync},\ \textsc{Ssync}\}$. And, also refines the computational landscape for the \textsc{Async} scheduler. In the Table~\ref{fig:compTable}, the full comparison is given. In our work the robots are \textit{disoriented}, i.e., it does not have an agreed notion of clockwise sense or any direction or handedness. The robots have full unobstructed visibility. The proofs of the Lemmas are omitted from the main paper and given in the Appendix Section~\ref{sec:appndx}.

 \begin{figure}[ht!]
    \centering
    \includegraphics[width=.99\linewidth]{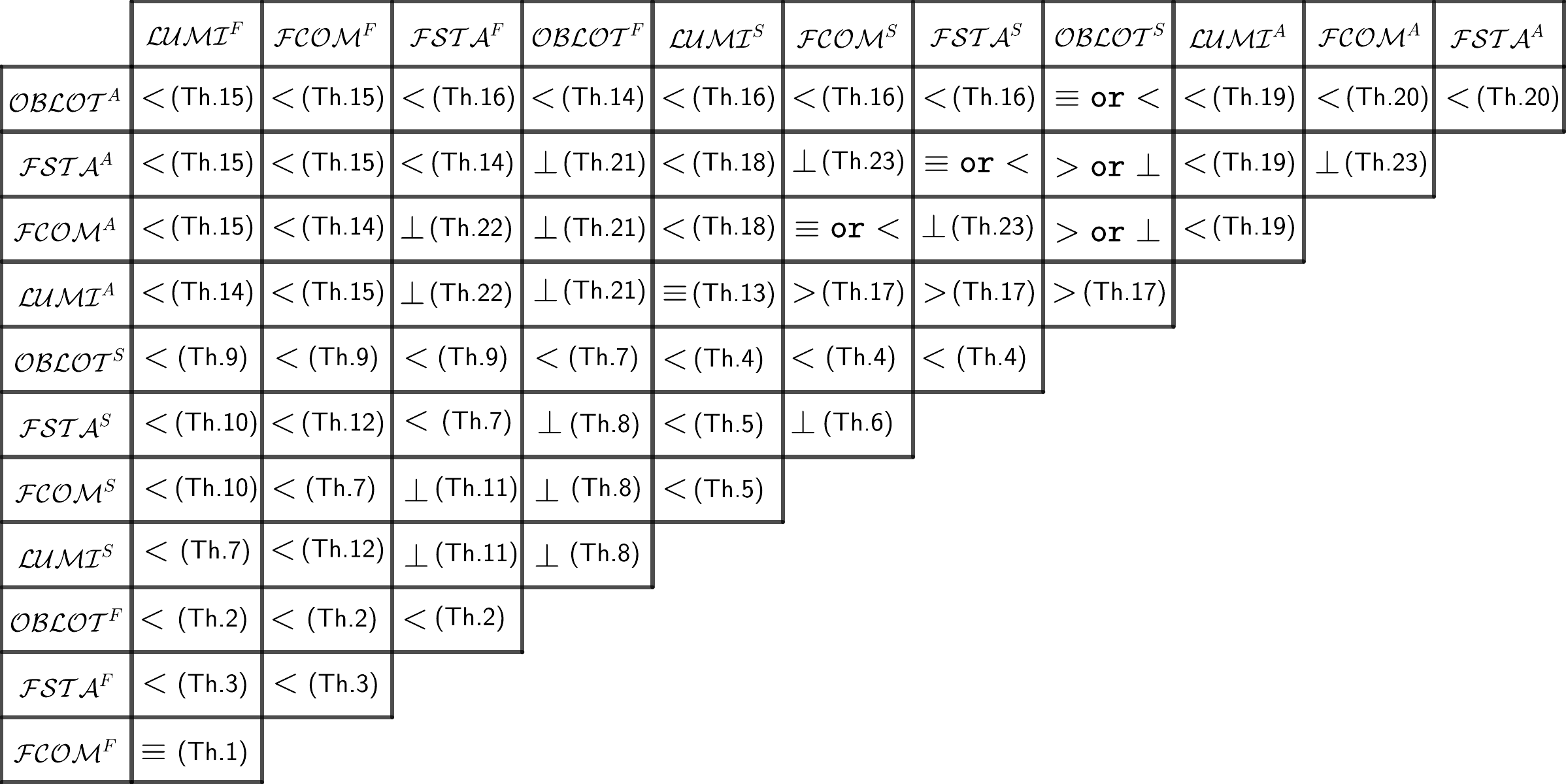}
    \caption{Comparison Table}
    \label{fig:compTable}
\end{figure}



\section{Model and Preliminaries}\label{model}
\subsection{Robots}
A set $R^{\mathcal{G}}=\{r_1,r_2,\dots,r_n\}$ of $n$ mobile computing units, called \textit{robots} operating on a graph $\mathcal{G}$ (can be finite or infinite) embedded on the euclidean plane. The robots are \textit{anonymous}, i.e., the robots do not have any unique identifier; \textit{identical}, i.e., the robots are indistinguishable from physical appearance; \textit{autonomous}, i.e., the robots have no central control; \textit{homogeneous}, i.e., the robots have same capabilities and run the same algorithm. A robot can rest on the nodes of the graphs, and from a node, a robot moves to another adjacent node. A robot cannot stop on an edge connecting to nodes. The robots are equipped with a local coordinate system where the robot considers itself at its origin and can find the position of other robots.

\paragraph*{\bf Look-Compute-Move cycle} The robots operate through a sequence of \textit{look-compute-move} (LCM) cycles. On activation, the robots enter into \textit{look} phase. In this phase, the robot takes an instantaneous snapshot of the surroundings. Thus, it finds out position and states (as declared) of other robots with respect to its local coordinate system. After that, the robot enters into \textit{compute} phase. In this phase, the robot runs an inbuilt algorithm that takes the information from the snapshot as input. Then it results in a destination position. The destination position is either an adjacent node of the node, or the node itself where the robot is currently situated at. In \textit{move} phase, the robot moves towards the computed destination position or, makes a null movement and stays put.

After the execution of an LCM cycle, a robot becomes inactive. Initially, all robots are inactive. An adversarial scheduler is assumed to be responsible for activating robots. The look phase is considered instantaneous and it is considered as the time instant of the snapshot. The processing and gathering of information from the snapshot is considered in the compute phase. The compute phase ends when the robot starts moving and enters into the move phase. The move phase includes the mechanical movement of the robot.

\paragraph*{\bf Considered Robot Models}
The weakest and classical robot model is $\mathcal{OBLOT}$. In this model, the robots are not further equipped with any technology that acts as a persistent memory or commutation architecture. Thus, in this model the robots are oblivious (absence of persistent memory) and silent (absence of communication ability). Hence, it cannot remember its past actions or past configurations from the previous LCM cycle. 

Another standard model in the literature is $\mathcal{LUMI}$. In this model, each robot $r$ is equipped with a register $Light[r]$ called \textit{light}. It can take values from a finite set $C$. Elements of the sets are called \textit{colors}. The color of the light of a robot is visible to other robots. A color can be set to the light at the end of the compute phase in an LCM cycle. The color of the light remains persistent in the next cycle and it does not get reset automatically at the end of a cycle. In the $\mathcal{LUMI}$ model the snapshot also consists of colors. Thus for each robot, it gets a pair (\textit{position, color}). A robot can have more than one light. For these cases, if the robot has $k$ lights, the color can be considered as an $k$ tuple ($r.light_1,\dots,r.light_k$) where $r.light_i$ is a color from $C$. Thus, it is equivalent to the case having one light with color set $C^k$ (cartesian product of $k$ number of $C$). Note that, in the $\mathcal{LUMI}$ model the set $C$ must have at least two colors. If $C$ has only one color then it is technically an $\mathcal{OBLOT}$ model.

There are two intermediate models of in between the above two which are $\mathcal{FSTA}$ and $\mathcal{FCOM}$. In the first model, the light is \textit{internal}. This means that the color of the light is only visible to the robot itself but not to other robots. Each color can be considered a different state of the robots. These robots are silent but they are finite-state. For the second model the the lights are \textit{external}. This means the colors of the light are visible to other robots but not to itself. This is used to communicate its color with other robots but it cannot see its own color in the next LCM cycle. Thus, the robots are oblivious but enabled with finite communication.

Here, a graph $\mathcal{G}$ is considered as an embedded graph, embedded on an Euclidean plane. Suppose a set of more than one robot is placed on $\mathcal{G}$. Let $f$ be a function from the set of vertices of $G$ to $\mathbb{N}\cup\{0\}$, where $f(v)$ is the number of robots on the vertex $v$ of $G$. Then the pair $(G,f)$ is said to be a \textit{configuration} of robots on $\mathcal{G}$.

\paragraph*{\bf Schedulers}

Depending on the activation schedule and time duration of LCM cycles of the robots there are main two types of schedulers. First, in a synchronous scheduler, the time is divided equally into rounds. In each round, activated robots simultaneously execute all the phases of the LCM cycles. In a fully-synchronous, (\textsc{Fsync}) scheduler in each round all the robots present in the system get activated. In a semi-asynchronous scheduler (\textsc{Ssync}), a nonempty subset of robots is activated in the scheduler and all activated robots  simultaneously execute all the phases of the LCM cycles. A \textit{fair} scheduler activates each robot infinitely often. We consider all schedulers to be fair. Next, in an asynchronous scheduler (\textsc{Async}) there is no common notion of time for robots. Each robot independently gets activated and executes its LCM cycle. The time length of LCM cycles, compute phases, and move phases of robots may be different. Even the length of two LCM cycles for one robot may be different. The gap between two consecutive LCM cycles, or the time length of an LCM cycle for a robot, is finite but can be unpredictably long. We consider the activation time and the time taken to complete an LCM cycle to be determined by the adversary.

\subsection{Problems and Computational Relationships}
Let $\mathcal{M}=\{\mathcal{OBLOT},\ \mathcal{FSTA},\ \mathcal{FCOM},\ \mathcal{LUMI}\}$ be the set of models considered in this work, and $\mathcal{S}=\{\textsc{Fsync},\ \textsc{Ssync},\ \textsc{Async}\}$ be the set of all considered schedulers. A \textit{variant} $M^K$ denotes a model $M$ together with scheduler $K$.  For convenience we shall denote $M^{\textsc{Fsync}}$, $M^{\textsc{Ssync}}$, and $M^{\textsc{Async}}$ respectively by $M^F$, $M^S$, and $M^A$.

A \textit{problem} (or, task) is described by a finite or infinite enumerable sequence of configurations together with some restrictions on the robots. The problem description implicitly gives a set of \textit{valid} initial configurations. An algorithm $\mathcal{A}$ is said to \textit{solve} a problem $P$ in a variant $M_K$ if starting from any valid initial configuration, any execution of $\mathcal{A}$ by a set of robots $R$ in the model $M$ under scheduler $K$ can form the configurations given in the sequence in the prescribed order maintaining the restrictions mentioned in $P$.

Let $M\in\mathcal{M}$ be a model and $K\in\mathcal{S}$, then $M(K)$ is the set of the problems that can be solved in model $M$ under then scheduler $K$. Let $M_1, M_2\in\mathcal{M}$ and $K_1,K_2\in\mathcal{S}$.

\begin{itemize}
    \item We say that a variant $M_1^{K_1}$ is \textit{computationally not less powerful} than a variant $M_2^{K_2}$, denoted as $M_1^{K_1}\ge M_2^{K_2}$, if $M_1(K_1)\supseteq M_2(K_2)$.
    
    \item We say that a variant $M_1^{K_1}$ is \textit{computationally more powerful} than a variant $M_2^{K_2}$, denoted as $M_1^{K_1}> M_2^{K_2}$, if $M_1(K_1)\supsetneq M_2(K_2)$.

    \item We say that a variant $M_1^{K_1}$ is \textit{computationally equivalent} than a variant $M_2^{K_2}$, denoted as $M_1^{K_1}\equiv M_2^{K_2}$, if $M_1(K_1)= M_2(K_2)$.

    \item We say that a variant $M_1^{K_1}$ is \textit{computationally orthogonal (or, incomparable)} than a variant $M_2^{K_2}$, denoted as $M_1^{K_1}\perp M_2^{K_2}$, if $M_1(K_1)\not\supseteq M_2(K_2)$ and $M_2(K_2)\not\supseteq M_1(K_1)$.

\end{itemize}

One can easy observe the following result.
\begin{lemma}\label{lm:trivial}
For any $M\in\mathcal{M}$ and $K\in\mathcal{S}$, 
(1)~$M^F\ge M^S \ge M^A$, (2)~$\mathcal{LUMI}^K\ge \mathcal{FSTA}^K\ge \mathcal{OBLOT}^K$, and (3)~$\mathcal{LUMI}^K\ge \mathcal{FCOM}^K\ge \mathcal{OBLOT}^K$.
\end{lemma}

\section{Comparisons in FSYNC}\label{comp_F}
This section compares the models $\mathcal{OBLOT}$, $\mathcal{FSTA}$, $\mathcal{FCOM}$ and $\mathcal{LUMI}$ in a fully-synchronous scheduler. 

\paragraph*{\bf Equivalency of $\mathcal{FCOM}^{F}$ and $\mathcal{LUMI}^{F}$:}
First, we show the equivalency of the models $\mathcal{FCOM}$ and $\mathcal{LUMI}$ in a fully synchronous scheduler. For that, we recall the similar result in \cite{flochinniOPODIS19} that deals with the same when the robots are operating in the euclidean plane. The proof given in \cite{flochinniOPODIS19} is independent of the fact that whether the robots are operating on a euclidean plane or a discrete domain. Hence we can conclude the following result.

\begin{theorem}\label{th:fcomLumi}
 $\mathcal{FCOM}^{F}\equiv\mathcal{LUMI}^{F}$.   
\end{theorem}

Next, let us define a problem named \textsc{moveOnce} below for further developments in this section.

\begin{definition}\label{def:moveOnce}
Let two robots $r_i\ (i=1,2)$ be placed on a vertex $v_i\ (i=1,2)$ of a five cycle with a chord such that $v_1$ and $v_2$ are not adjacent, and degree of each $v_i$ is 2 (Fig.~\ref{fig:moveOnce}). The problem \textsc{moveOnce} asks (1)~the robot, that has an adjacent node with degree 2, to move once to that adjacent node and stay still afterward, (2)~other robot will remain still forever.   
\end{definition}

\begin{figure}[ht!]
    \centering
    \includegraphics[width=.17\linewidth]{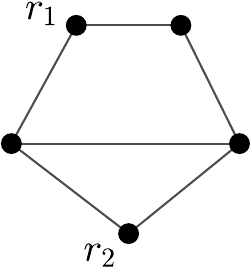}
    \caption{An image related to problem \textsc{moveOnce}}
    \label{fig:moveOnce}
\end{figure}


\begin{lemma}\label{lm:MCoblot}
  $\exists R\in\mathcal{R}_2,\ \textsc{moveOnce}\not 
\in \mathcal{OBLOT}^F(R)$.    
\end{lemma}

\begin{lemma}\label{lm:MCfsta}
  $\forall R\in\mathcal{R}_2,\ \textsc{moveOnce} 
\in \mathcal{FSTA}^A(R)$.    
\end{lemma}

\begin{theorem}\label{th:oblotf<fstaf}
    $\mathcal{OBLOT}^F < M^F$ for all $M\in\{\mathcal{FSTA,\ FCOM,\ LUMI}\}$.
\end{theorem}
\begin{proof}
    lemma~\ref{lm:MCoblot} and Lemma~\ref{lm:MCfsta} show that there exists a problem (\textsc{moveOnce}) that is solvable in $\mathcal{FSTA}^F$, thus solvable in $\mathcal{LUMI}^F$, but not solvable in $\mathcal{OBLOT}^F$. The problem \textsc{moveOnce} is solvable in $\mathcal{FCOM}^F$ from Theorem~\ref{th:fcomLumi}. Hence the result follows.
\end{proof}

 Next, introduce a new problem by modifying the problem $\neg IL$ defined in \cite{flochinniOPODIS19} for the discrete version.
 
\begin{definition}\label{def:-il}
    Let three robots be placed on a geometric graph given in Fig.~\ref{fig:il}. The initial configuration formed by robots is Config-I (Fig.~\ref{fig:il1}). In problem $\neg IL$ the robots are required to move to form Config-II (Fig.~\ref{fig:il2}).  from Config-I, then Config-III (Fig.~\ref{fig:il3}) from Config-II.   
\end{definition}

\begin{figure}[ht!]
    \centering
    \includegraphics[width=.15\linewidth]{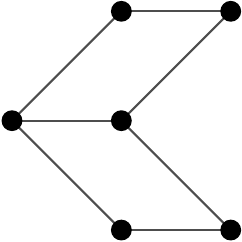}
    \caption{An image related to problem \textsc{$\neg IL$}}
    \label{fig:il}
\end{figure}


\begin{figure}[ht!]
\begin{minipage}{.3\textwidth}
  \centering
  \includegraphics[width=.6\linewidth]{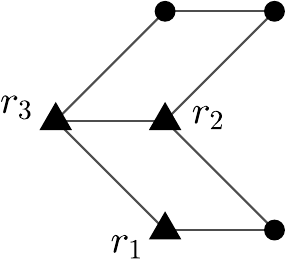}
  \caption{Config-I}
  \label{fig:il1}
\end{minipage}
\begin{minipage}{.3\textwidth}
  \centering
  \includegraphics[width=.5\linewidth]{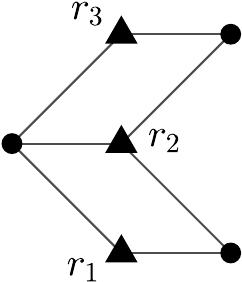}
  \caption{Config-II}
  \label{fig:il2}
\end{minipage}
\begin{minipage}{.3\textwidth}
  \centering
  \includegraphics[width=.57\linewidth]{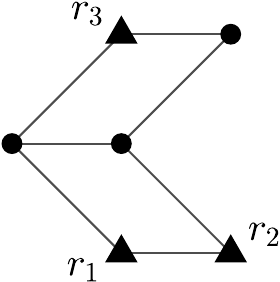}
  \caption{Config-III}
  \label{fig:il3}
\end{minipage}
\end{figure}

The following Lemma~\ref{lm:fstaIL} and Lemma~\ref{lm:fcomIL} follow from the same arguments as given in \cite{flochinniOPODIS19}. Let us recall the arguments briefly in the next two lines.  In the $\mathcal{FSTA}$ model it is impossible for $r_2$ to understand which robot has moved to form Config-II. But in $\mathcal{FCOM}$ model the moving robot can set a distinct color in order to let $r_2$ know it. 

\begin{lemma}\label{lm:fstaIL}
  $\exists R\in\mathcal{R}_3,\ \neg IL\not 
\in \mathcal{FSTA}^F(R)$.    
\end{lemma}

\begin{lemma}\label{lm:fcomIL}
  $\forall R\in\mathcal{R}_3,\ \neg IL 
\in \mathcal{FCOM}^A(R)$.    
\end{lemma}


\begin{theorem}\label{th:fstaFvsallF}
  $\mathcal{FSTA}^{F}<M^{F}$, for all $M\in \{\mathcal{FCOM,\ LUMI}\}$.
\end{theorem}
\begin{proof}
    Lemma~\ref{lm:fstaIL} and Lemma~\ref{lm:fcomIL} show that the $\neg IL$ problem is solvable in $\mathcal{FCOM}^F$ but not solvable in $\mathcal{FSTA}^F$. Next, from Theorem~\ref{th:fcomLumi} we have $\mathcal{FCOM}^F$ and $\mathcal{LUMI}^F$ are equivalent. Thus, the result is proven.
\end{proof}
 
\section{Comparisons in SSYNC}\label{comp_S}
This section compares the models $\mathcal{OBLOT}$, $\mathcal{FSTA}$, $\mathcal{FCOM}$, and $\mathcal{LUMI}$ in a semi-synchronous scheduler. 


\begin{theorem}\label{th:oblotS<allS}
   $\mathcal{OBLOT}^{S}<M^{S}$, for all $M\in\{\mathcal{FSTA,\ FCOM,\ LUMI}\}$.
\end{theorem}
\begin{proof}
    In Lemma~\ref{lm:MCoblot} and Lemma~\ref{lm:MCfsta}, we showed that there exists a problem solvable in $\mathcal{FSTA}^S$, thus also in $\mathcal{LUMI}^S$, but not solvable in $\mathcal{OBLOT}^S$. Also, in Lemma~\ref{lm:fstaIL} and Lemma~\ref{lm:fcomIL}, we showed that there exists a problem solvable in $\mathcal{FCOM}^S$ but solvable in $\mathcal{OBLOT}^S$. Hence the result follows.
\end{proof}

In Lemma~\ref{lm:MCoblot}, we have shown that \textsc{moveOnce} is unsolvable in $\mathcal{OBLOT}^F$. In the following Lemma~\ref{lm:MCfcom}, we show that \textsc{moveOnce} is unsolvable in $\mathcal{FCOM}^S$.

\begin{lemma}\label{lm:MCfcom}
  $\exists R\in\mathcal{R}_2,\ \textsc{moveOnce}\not 
\in \mathcal{FCOM}^S(R)$.    
\end{lemma}


\begin{theorem}\label{th:allS<lumiS}
    $M^{S}<\mathcal{LUMI}^{S}$ for all  $M\in\{\mathcal{FSTA,\ FCOM}\}$.
\end{theorem}
\begin{proof}
    In Lemma~\ref{lm:fstaIL} and Lemma~\ref{lm:fcomIL} we showed that there exists a problem that is solvable in $\mathcal{LUMI}^S$ but not solvable in $\mathcal{FSTA}^S$. Then, in Lemma~\ref{lm:MCfsta} and Lemma~\ref{lm:MCfcom} we showed that there exists a problem that is solvable in $\mathcal{LUMI}^S$ but not solvable in $\mathcal{FCOM}^S$. Hence the result follows.
\end{proof}

\paragraph*{\bf Orthogonality of $\mathcal{FSTA}^{S}$ and $\mathcal{FCOM}^{S}$}
We define a new problem named OSP in the following definition.

\begin{definition}\label{def:OSP}
    Suppose, three robots are placed at distinct positions on an infinite discrete line having Configuration $\mathcal{A}$ (Fig.~\ref{fig:C1}). The OSP problem asks the robots to move such that it creates the following sequence of configurations:
    <$\mathcal{B}$, $\mathcal{C}$, $\mathcal{B}$, $\mathcal{A}$, $\mathcal{B}$, $\mathcal{C}$, $\mathcal{B}$, $\mathcal{A}$, ...  >, where configuration $\mathcal{B}$ and $\mathcal{C}$ are given in Fig.~\ref{fig:C2} and Fig.~\ref{fig:C3} respectively.
\end{definition}

 \begin{figure}[ht!]
         \centering
         \includegraphics[width=0.65\textwidth]{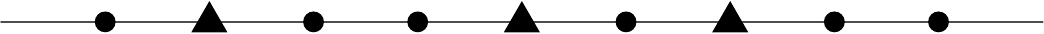}
         \caption{Configuration $\mathcal{A}$}
         \label{fig:C1}
     \end{figure}
     \begin{figure}[ht!]
         \centering
         \includegraphics[width=0.65\textwidth]{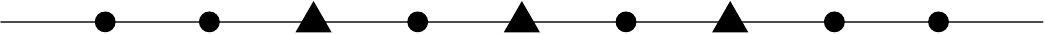}
         \caption{Configuration $\mathcal{B}$}
         \label{fig:C2}
     \end{figure}
     \begin{figure}[ht!]
         \centering
         \includegraphics[width=0.65\textwidth]{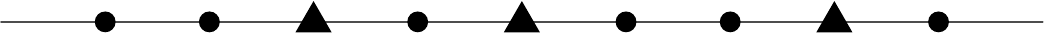}
         \caption{Configuration $\mathcal{C}$}
         \label{fig:C3}
     \end{figure}

We show that this problem cannot be solved under semi-synchronous scheduler.

\begin{lemma}\label{lm:fstaOSP}
  $\exists R\in\mathcal{R}_3,\ \textsc{OSP}\not 
\in \mathcal{FSTA}^S(R)$.    
\end{lemma}

Next, we present an algorithm \textsc{AlgoOSP} that solves the OSP in $\mathcal{FCOM}^S$. Suppose each robot has an external light that is initially set \texttt{off} color and can take the color \texttt{colN}, \texttt{colF}, and \texttt{colT}. The middle robot does nothing and its color remains \texttt{off}. We denote the color variable of a robot $r$ as $color.r$.

\begin{algorithm}[ht!]
\caption{\footnotesize\textsc{\textsc{AlgoOSP}}: for a generic robot $r$ where $r'$ is the other terminal robot; $d$ and $d'$ are respectively the distance of $r$ and $r'$  from middle robot}\label{algo:OSP} 
\footnotesize
    \uIf{$color.r'$=\texttt{off}}
    {
        \uIf{$d>d'$}
            {
                $color.r=$ \texttt{colN}\;
                $r$ moves towards the middle robot\;
            }
        \ElseIf{$d<d'$}
            {
                $color.r=$ \texttt{off}\;
            }
    }
    \uElseIf{$color.r'$ = \texttt{colN}}
    {
        \uIf{$d=d'$}
        {
            $color.r=$ \texttt{colN}\;
            $r$ moves away from the middle robot\;
        }
        \ElseIf{$d<d'$}
        {
            $color.r=$ \texttt{colF}\;    
        }
    }
    \uElseIf{$color.r'$ = \texttt{colF} and $d> d'$ }
    {
        $color.r=$ \texttt{colT}\; 
        $r$ moves towards the middle robot\;
    }
    \ElseIf{$color.r'$ = \texttt{colT} and $d=d'$}
    {
        $r$ moves away from the middle robot\;
        $color.r=$ \texttt{off}\; 
    }
    \end{algorithm}

\begin{lemma}\label{lm:fcomOSP}
    The \textsc{AlgoOSP} solves the OSP in model $\mathcal{FCOM}^S$.
\end{lemma}

\begin{theorem}\label{th:fstaSvs.FcomS}
$\mathcal{FSTA}^{S}\perp\mathcal{FCOM}^{S}$.    
\end{theorem}
\begin{proof}
    In Lemma~\ref{lm:MCfsta} and Lemma~\ref{lm:MCfcom}, we showed that there exists a problem that is solvable in $\mathcal{FSTA}^S$ but not solvable in $\mathcal{FCOM}^S$. In  Lemma~\ref{lm:fstaOSP} and Lemma~\ref{lm:fcomOSP} we showed that there exists a problem that is solvable in $\mathcal{FCOM}^S$ but not solvable in $\mathcal{FSTA}^S$. Thus the result follows.
\end{proof}





\section{Cross-comparison in FSYNC and SSYNC}\label{comp_FS}

First, we define a modified version of the Stand Up Indulgent Rendezvous (SUIR) problem in the following.

\begin{definition}\label{def:suir}
  Let two robots be placed on the endpoints of a line graph of three nodes (Fig.~\ref{fig:suir}). One of the robots may crash (i.e., stop working completely) at any time, and it does not activate ever after. The problem SUIR asks the robots to meet at the middle node (node with degree two) unless one of the robots crashes before gathering. Otherwise, the non-crashed robot moves to the crashed robot's position. A robot cannot identify whether a robot has crashed or not from the snapshot. 
\end{definition}

\begin{figure}[ht!]
    \centering
    \includegraphics[width=.18\linewidth]{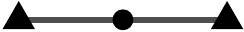}
    \caption{An image regarding problem SIUR}
    \label{fig:suir}
\end{figure}

 \begin{lemma}\label{lm:SUIRlumi}
    $\exists R\in\mathcal{R}_2, \ SUIR
\not\in \mathcal{LUMI}^S(R)$.
\end{lemma}

\begin{lemma}\label{lm:SUIRoblot}
    $\forall R\in\mathcal{R}_2, \ SUIR
\in \mathcal{OBLOT}^F(R)$.
\end{lemma}

\begin{theorem}\label{th:allS<allF}
    $M^{S}<M^{F}$ for all $M\in\{\mathcal{OBLOT,\ FSTA,\ FCOM,\ LUMI}\}$.
\end{theorem}
\begin{proof}
    In Lemma~\ref{lm:SUIRlumi} and Lemma~\ref{lm:SUIRoblot} we showed that the problem SUIR is solvable in $M^F$ for all $M\in\{\mathcal{OBLOT,\ FSTA,\ FCOM,\ LUMI}\}$ but not solvable in $M^S$ for all $M\in\{\mathcal{OBLOT,\ FSTA,\ FCOM,\ LUMI}\}$. Thus, the result follows.
\end{proof}

\begin{theorem}\label{th:oblotFvs.allS}
     $\mathcal{OBLOT}^{F}\perp M^{S}$ for all $M\in\{\mathcal{FSTA,\ FCOM,\ LUMI}\}$.
 \end{theorem}
\begin{proof}
    In Lemma~\ref{lm:SUIRlumi} and Lemma~\ref{lm:SUIRoblot}, we showed that the problem SUIR is solvable in $\mathcal{OBLOT}^F$ but not solvable in $M^S$ for all $M\in\{\mathcal{FSTA,\ FCOM,\ LUMI}\}$. In Lemma~\ref{lm:MCoblot} and Lemma~\ref{lm:MCfsta} we showed that there exists a problem that solves a problem in $\mathcal{FSTA}^S$ and thus in $\mathcal{LUMI}^S$ but not solvable in $\mathcal{OBLOT}^F$. In Lemma~\ref{lm:fstaIL} and Lemma~\ref{lm:fcomIL}, we showed that there exists a problem that is solvable in $\mathcal{FCOM}^S$ but not solvable in $\mathcal{OBLOT}^F$. Thus, the result follows. 
\end{proof}

 \begin{theorem}\label{th:allFvs.oblotS}
     $ \mathcal{OBLOT}^{S}<M^{F}$ for all $M\in\{\mathcal{FSTA,\ FCOM,\ LUMI}\}$.
 \end{theorem} 
 \begin{proof}
     From Theorem~\ref{th:oblotS<allS} this result follows.
 \end{proof}

 \begin{theorem}\label{th:lumiF>allS}
    $M^{S}<\mathcal{LUMI}^{F}$ for all $M\in\{\mathcal{FSTA,\ FCOM}\}$.
\end{theorem}
\begin{proof}
    $\mathcal{FSTA}^{S}<\mathcal{LUMI}^{F}$ follows from Theorem~\ref{th:fstaFvsallF} and $\mathcal{FCOM}^{S}<\mathcal{LUMI}^{F}$ follows from Theorem~\ref{th:allS<allF}.
\end{proof}

\begin{theorem}\label{th:fstaFvs.allS}
     $\mathcal{FSTA}^{F}\perp M^{S}$ for all $M\in\{\mathcal{FCOM,\ LUMI}\}$.
 \end{theorem}
 \begin{proof}
     In Lemma~\ref{lm:fstaIL} and Lemma~\ref{lm:fcomIL}, we showed that there is a problem that is solvable in $\mathcal{FCOM}^S$, thus in $\mathcal{LUMI}^S$, but not solvable in $\mathcal{FSTA}^F$. In Lemma~\ref{lm:SUIRlumi} and Lemma~\ref{lm:SUIRoblot}, we showed that there is a problem that is solvable in $\mathcal{FSTA}^F$ but not solvable in $\mathcal{LUMI}^S$, thus not in $\mathcal{FCOM}^S$. Thus, the result follows.
 \end{proof}

 \begin{theorem}\label{th:fcomFvs.allS}
     $\mathcal{FCOM}^{F}> M^{S}$ for all $M\in\{\mathcal{FSTA,\ LUMI}\}$.
 \end{theorem}
 \begin{proof}
    Since from Theorem~\ref{th:fcomLumi} we have $\mathcal{FCOM}^{F}$ and $\mathcal{LUMI}^{F}$ are equivalent computationally. Therefore we need to show $\mathcal{LUMI}^{F}> M^{S}$ for all $M\in\{\mathcal{FSTA,\ LUMI}\}$. $\mathcal{LUMI}^{F}> \mathcal{FSTA}^{S}$ follows from Theorem~\ref{th:lumiF>allS} and $\mathcal{LUMI}^{F}> \mathcal{LUMI}^{S}$ follows from Theorem~\ref{th:allS<allF}.
 \end{proof}

\section{Comparisons under ASYNC scheduler}\label{sec:Async}
 In this section we investigate comparisons when one variant has scheduler \textsc{Fsync}. From a result of \cite{DSFN18}, we have the following Theorem~\ref{th:lumiSlumiA}.

 \begin{theorem}[\cite{DSFN18}]\label{th:lumiSlumiA}
     $\mathcal{LUMI}^{S}\equiv\mathcal{LUMI}^{A}$.
 \end{theorem}

 \begin{theorem}\label{th:allAsyncvs.allFsync}
     $M^{A}<M^{F}$ for all $M\in\mathcal{M}$.
 \end{theorem}
 \begin{proof}
     From Theorem~\ref{th:allS<allF}, the result follows.
 \end{proof}
 \begin{theorem}\label{th:allAvs.lumiF}
     $M^{A}<M_1^{F}$ for all $M\in\mathcal{M}$ and $M_1\in\{\mathcal{LUMI},\ \mathcal{FCOM}\}$.
 \end{theorem}
 \begin{proof}
     From Theorem~\ref{th:allS<allF}, Theorem~\ref{th:allFvs.oblotS}, and Theorem~\ref{th:lumiF>allS}, $M^{A}<M_1^{F}$ for all $M\in\mathcal{M}$ and $M^{A}<\mathcal{LUMI}^{F}$ follows. From Theorem~\ref{th:fcomLumi}, $\mathcal{LUMI}^F\equiv\mathcal{FCOM}^F$, the rest part of the result follows.
 \end{proof}
 \begin{theorem}\label{th:oblotA<allSF}
   $\mathcal{OBLOT}^{A}<M^{K}$, for all $M\in\{\mathcal{FSTA,\ FCOM,\ LUMI}\}$ and $K\in\{\textsc{Ssync},\textsc{Fsync}\}$.
\end{theorem}
\begin{proof}
    From Theorem~\ref{th:oblotS<allS} the result follows.
\end{proof}
\begin{theorem}\label{th:lumiAvs.allS}
    $\mathcal{LUMI}^{A}>M^{S}$ for all $M\in\{\mathcal{FSTA,\ FCOM,\ OBLOT}\}$. 
\end{theorem}
\begin{proof}
    From Theorem~\ref{th:oblotS<allS} and Theorem~\ref{th:allS<lumiS}, we can have $\mathcal{LUMI}^{S}>M^{S}$ for all $M\in\{\mathcal{FSTA,\ FCOM,\ OBLOT}\}$. Then from Theorem~\ref{th:lumiSlumiA}, the result follows.
\end{proof}
\begin{theorem}\label{th:allA<lumiS}
    $M^{A}<\mathcal{LUMI}^{S}$ for all  $M\in\{\mathcal{FSTA,\ FCOM}\}$.
\end{theorem}
\begin{proof}
    From Theorem~\ref{th:allS<lumiS}, the result follows.
\end{proof}
\begin{theorem}\label{th:lumiA>allA}
    $M^{A}<\mathcal{LUMI}^{A}$ for all  $M\in\{\mathcal{OBLOT,\ FSTA,\ FCOM}\}$.
\end{theorem}
\begin{proof}
       In Lemma~\ref{lm:fstaIL} and Lemma~\ref{lm:fcomIL} we showed that there exists a problem that is solvable in $\mathcal{LUMI}^A$ but not solvable in $\mathcal{FSTA}^A$, thus in $\mathcal{OBLOT}^A$. Then, in Lemma~\ref{lm:MCfsta} and Lemma~\ref{lm:MCfcom} we showed that there exists a problem that is solvable in $\mathcal{LUMI}^A$ but not solvable in $\mathcal{FCOM}^A$. Hence the result follows.
\end{proof}
\begin{theorem}\label{th:oblotA<allA}
    $\mathcal{OBLOT}^{A}<M^{A}$ for all  $M\in\{\mathcal{FSTA,\ FCOM}\}$.
\end{theorem}
\begin{proof}
    In Lemma~\ref{lm:fstaIL} and Lemma~\ref{lm:fcomIL} we showed that there exists a problem that is solvable in $\mathcal{FCOM}^A$ but not solvable in $\mathcal{OBLOT}^A$. Then, in Lemma~\ref{lm:MCfsta} and Lemma~\ref{lm:MCfcom} we showed that there exists a problem that is solvable in $\mathcal{FSTA}^A$ but not solvable in $\mathcal{OBLOT}^A$. Hence the result follows.
\end{proof}

\begin{theorem}\label{th:oblotFvs.allA}
     $\mathcal{OBLOT}^{F}\perp M^{A}$ for all $M\in\{\mathcal{FSTA,\ FCOM,\ LUMI}\}$.
 \end{theorem}
\begin{proof}
    In Lemma~\ref{lm:SUIRlumi} and Lemma~\ref{lm:SUIRoblot}, we showed that the problem SUIR is solvable in $\mathcal{OBLOT}^F$ but not solvable in $M^A$ for all $M\in\{\mathcal{FSTA,\ FCOM,\ LUMI}\}$. In Lemma~\ref{lm:MCoblot} and Lemma~\ref{lm:MCfsta}, we showed that there exists a problem that solves a problem in $\mathcal{FSTA}^A$ and thus in $\mathcal{LUMI}^A$ but not solvable in $\mathcal{OBLOT}^F$. In Lemma~\ref{lm:fstaIL} and Lemma~\ref{lm:fcomIL}, we showed that there exists a problem that is solvable in $\mathcal{FCOM}^A$ but not solvable in $\mathcal{OBLOT}^F$. Thus, the result follows. 
\end{proof}

 \begin{theorem}\label{th:fstaFvs.allA}
     $\mathcal{FSTA}^{F}\perp M^{A}$ for all $M\in\{\mathcal{FCOM,\ LUMI}\}$.
 \end{theorem}
 \begin{proof}
     In Lemma~\ref{lm:fstaIL} and Lemma~\ref{lm:fcomIL}, we showed that there is a problem that is solvable in $\mathcal{FCOM}^A$, thus in $\mathcal{LUMI}^A$, but not solvable in $\mathcal{FSTA}^F$. In Lemma~\ref{lm:SUIRlumi} and Lemma~\ref{lm:SUIRoblot}, we showed that there is a problem that is solvable in $\mathcal{FSTA}^F$ but not solvable in $\mathcal{LUMI}^A$, thus not in $\mathcal{FCOM}^A$. Thus, the result follows.
 \end{proof}

 \begin{theorem}\label{th:fcomallvs.fstaall}
     $\mathcal{FSTA}^{K_1}\perp\mathcal{FCOM}^{K_2}$ for all $K_1,K_2\in\{\textsc{Async},\textsc{Ssync}\}$.
 \end{theorem}
 \begin{proof}
If both $K_1$ and $K_2$ are \textsc{Ssync}, then it is Theorem~\ref{th:fstaSvs.FcomS}. In Lemma~\ref{lm:MCfsta} and Lemma~\ref{lm:MCfcom}, we showed that there exists a problem that is solvable in $\mathcal{FSTA}^A$ but not solvable in $\mathcal{FCOM}^S$. In  Lemma~\ref{lm:fstaOSP} and Lemma~\ref{lm:fcomOSP} we showed that there exists a problem that is solvable in $\mathcal{FCOM}^A$ but not solvable in $\mathcal{FSTA}^S$. Thus, the rest of the result follows.  
 \end{proof}


Finally, we achieve the table in the Figure~\ref{fig:compTable}. In the next section, we conclude the work.

 \section{Conclusion}\label{conclusion}
In this work, a complete computational relationships between four different models $\mathcal{OBLOT}$, $\mathcal{FSTA}$, $\mathcal{FCOM}$, and $\mathcal{LUMI}$ under two synchronous schedulers \textsc{Fsync} and \textsc{Ssync} when a set of anonymous, identical autonomous, homogeneous robots are operating on a graph embedded on euclidean plane. Also, this work refines computational landscape for \textsc{Async} scheduler. A series of works has been dedicated to refine the computational landscape when robots are operating on an euclidean plane, but the same was missing when robots are operating on discrete regions like on different graphs. This work aims to fill the gap. The table in Figure~\ref{fig:compTable} provides the results obtains in this work. Some of the comparisons in \textsc{Async} scheduler are still remain open.

\bibliographystyle{splncs04}
\bibliography{model_comp}

\begin{thebibliography}{10}
\providecommand{\url}[1]{\texttt{#1}}
\providecommand{\urlprefix}{URL }
\providecommand{\doi}[1]{https://doi.org/#1}

\bibitem{BKLT24}
Bramas, Q., Kamei, S., Lamani, A., Tixeuil, S.: Stand-up indulgent gathering on rings. In: Emek, Y. (ed.) Structural Information and Communication Complexity - 31st International Colloquium, {SIROCCO} 2024, Vietri sul Mare, Italy, May 27-29, 2024, Proceedings. Lecture Notes in Computer Science, vol. 14662, pp. 119--137. Springer (2024). \doi{10.1007/978-3-031-60603-8\_7}

\bibitem{BFKPSW21}
Buchin, K., Flocchini, P., Kostitsyna, I., Peters, T., Santoro, N., Wada, K.: Autonomous mobile robots: Refining the computational landscape. In: {IEEE} International Parallel and Distributed Processing Symposium Workshops, {IPDPS} Workshops 2021, Portland, OR, USA, June 17-21, 2021. pp. 576--585. {IEEE} (2021). \doi{10.1109/IPDPSW52791.2021.00091}

\bibitem{BFKPSW22}
Buchin, K., Flocchini, P., Kostitsyna, I., Peters, T., Santoro, N., Wada, K.: On the computational power of energy-constrained mobile robots: Algorithms and cross-model analysis. In: Parter, M. (ed.) Structural Information and Communication Complexity - 29th International Colloquium, {SIROCCO} 2022, Paderborn, Germany, June 27-29, 2022, Proceedings. Lecture Notes in Computer Science, vol. 13298, pp. 42--61. Springer (2022). \doi{10.1007/978-3-031-09993-9\_3}

\bibitem{DGSGS24}
Das, A., Ghosh, S., Sharma, A., Goswami, P., Sau, B.: The computational landscape of autonomous mobile robots: The visibility perspective. In: Devismes, S., Mandal, P.S., Saradhi, V.V., Prasad, B., Molla, A.R., Sharma, G. (eds.) Distributed Computing and Intelligent Technology - 20th International Conference, {ICDCIT} 2024, Bhubaneswar, India, January 17-20, 2024, Proceedings. Lecture Notes in Computer Science, vol. 14501, pp. 85--100. Springer (2024). \doi{10.1007/978-3-031-50583-6\_6}

\bibitem{FPSY16}
Das, S., Flocchini, P., Prencipe, G., Santoro, N., Yamashita, M.: Autonomous mobile robots with lights. Theor. Comput. Sci.  \textbf{609},  171--184 (2016). \doi{10.1016/J.TCS.2015.09.018}

\bibitem{DSFN18}
D'Emidio, M., Stefano, G.D., Frigioni, D., Navarra, A.: Characterizing the computational power of mobile robots on graphs and implications for the euclidean plane. Inf. Comput.  \textbf{263},  57--74 (2018). \doi{10.1016/J.IC.2018.09.010}

\bibitem{FMMP24}
Feletti, C., Mambretti, L., Mereghetti, C., Palano, B.: {Computational Power of Opaque Robots}. In: Casteigts, A., Kuhn, F. (eds.) 3rd Symposium on Algorithmic Foundations of Dynamic Networks (SAND 2024). Leibniz International Proceedings in Informatics (LIPIcs), vol.~292, pp. 13:1--13:19. Schloss Dagstuhl -- Leibniz-Zentrum f{\"u}r Informatik, Dagstuhl, Germany (2024). \doi{10.4230/LIPIcs.SAND.2024.13}

\bibitem{FPSY99}
Flocchini, P., Prencipe, G., Santoro, N., Widmayer, P.: Hard tasks for weak robots: The role of common knowledge in pattern formation by autonomous mobile robots. In: Proceedings of the 10th International Symposium on Algorithms and Computation. p. 93–102. ISAAC '99, Springer-Verlag, Berlin, Heidelberg (1999)

\bibitem{FSS023}
Flocchini, P., Santoro, N., Sudo, Y., Wada, K.: On asynchrony, memory, and communication: Separations and landscapes. In: Bessani, A., D{\'{e}}fago, X., Nakamura, J., Wada, K., Yamauchi, Y. (eds.) 27th International Conference on Principles of Distributed Systems, {OPODIS} 2023, December 6-8, 2023, Tokyo, Japan. LIPIcs, vol.~286, pp. 28:1--28:23. Schloss Dagstuhl - Leibniz-Zentrum f{\"{u}}r Informatik (2023). \doi{10.4230/LIPICS.OPODIS.2023.28}

\bibitem{FSVY160}
Flocchini, P., Santoro, N., Viglietta, G., Yamashita, M.: Rendezvous with constant memory. Theor. Comput. Sci.  \textbf{621},  57--72 (2016). \doi{10.1016/J.TCS.2016.01.025}

\bibitem{FSW19}
Flocchini, P., Santoro, N., Wada, K.: On memory, communication, and synchronous schedulers when moving and computing. In: Felber, P., Friedman, R., Gilbert, S., Miller, A. (eds.) 23rd International Conference on Principles of Distributed Systems, {OPODIS} 2019, December 17-19, 2019, Neuch{\^{a}}tel, Switzerland. LIPIcs, vol.~153, pp. 25:1--25:17. Schloss Dagstuhl - Leibniz-Zentrum f{\"{u}}r Informatik (2019). \doi{10.4230/LIPICS.OPODIS.2019.25}

\bibitem{flochinniOPODIS19}
Flocchini, P., Santoro, N., Wada, K.: {On Memory, Communication, and Synchronous Schedulers When Moving and Computing}. In: Felber, P., Friedman, R., Gilbert, S., Miller, A. (eds.) 23rd International Conference on Principles of Distributed Systems (OPODIS 2019). Leibniz International Proceedings in Informatics (LIPIcs), vol.~153, pp. 25:1--25:17. Schloss Dagstuhl -- Leibniz-Zentrum f{\"u}r Informatik, Dagstuhl, Germany (2020). \doi{10.4230/LIPIcs.OPODIS.2019.25}

\bibitem{KKNPS21}
Kirkpatrick, D.G., Kostitsyna, I., Navarra, A., Prencipe, G., Santoro, N.: Separating bounded and unbounded asynchrony for autonomous robots: Point convergence with limited visibility. In: Miller, A., Censor{-}Hillel, K., Korhonen, J.H. (eds.) {PODC} '21: {ACM} Symposium on Principles of Distributed Computing, Virtual Event, Italy, July 26-30, 2021. pp. 9--19. {ACM} (2021). \doi{10.1145/3465084.3467910}

\bibitem{PP24}
Pattanayak, D., Pelc, A.: Deterministic treasure hunt and rendezvous in arbitrary connected graphs. Inf. Process. Lett.  \textbf{185},  106455 (2024). \doi{10.1016/J.IPL.2023.106455}

\bibitem{SGGS24}
Sharma, A., Ghosh, S., Goswami, P., Sau, B.: Space and move-optimal arbitrary pattern formation on a rectangular grid by robot swarms. In: Proceedings of the 25th International Conference on Distributed Computing and Networking, {ICDCN} 2024, Chennai, India, January 4-7, 2024. pp. 65--73. {ACM} (2024). \doi{10.1145/3631461.3631542}

\bibitem{SuzukiY96}
Suzuki, I., Yamashita, M.: Distributed anonymous mobile robots. In: Santoro, N., Spirakis, P.G. (eds.) SIROCCO'96, The 3rd International Colloquium on Structural Information {\&} Communication Complexity, Siena, Italy, June 6-8, 1996. pp. 313--330. Carleton Scientific (1996)

\end{thebibliography}

\newpage
\section{APPENDIX}\label{sec:appndx}

\paragraph{\bf Proof of Lemma~\ref{lm:MCoblot}:}
If possible let there exist an algorithm $\mathcal{A}$ that solves the problem in $\mathcal {OBLOT}^{F}$. Since the robots are oblivious and have no communication ability, so after finishing one move of robot $r$, the configuration remains the same as the initial one. Thus, on activation, the robot $r$ will again move back to its initial position, which is a contradiction.

\paragraph{\bf Proof of Lemma~\ref{lm:MCfsta}:}
Let the robots have two states: \texttt{off} and \texttt{done}. Suppose initially both robots have state \texttt{off}. From the graph topology, the robots can identify which one to move. Let $r_1$ be the robot that is supposed to move. Then $r_2$ does not move at all. Then on activation first time (when its state is set as \texttt{off}) $r_1$ shall move changing its state to \texttt{done}. Next time, on activation it sees its state is set as \texttt{done}, so it does nothing. Even after the move of $r_1$ is done, the robot $r_2$ can identify itself as the robot that is not supposed to move. Thus, the above-described algorithm successfully solves \textsc{moveOnce}.

\paragraph{\bf Proof of Lemma~\ref{lm:MCfcom}:}
Let $R=\{r_1,r_2\}$. If possible let there exist an algorithm $\mathcal{A}$ that solves the problem in $\mathcal{FCOM}^S$. Let $r_2$ be the robot that needs to move. Let both robots be activated at every round. Let at $k^{th}$ round $r_1$ first time moves. Consider another execution where from first to $(k-1)^{th}$ round, both robots get activated, but in $k^{th}$, $r_1$ gets activated but $r_2$ is not. Since the robots have no shared notion of handedness or charity, the view of $r_1$ remains the same at the start of the $(k+1)^{th}$ round. Thus, in this execution, $r_1$ will again move to its initial position, which contradicts the correctness of the $\mathcal{A}$.

\paragraph{\bf Proof of Lemma~\ref{lm:fstaOSP}:}
If possible there is an algorithm $\mathsf{A}$ that solves the problem in a semi-synchronous scheduler. Consider an execution $E$ where all the robots are activated in each round. Then there will be rounds $k_1$, $k_2$, $k_3$, and $k_4$ such that $(k_1<k_2<k_3<k_4)$ and the following holds true.
\begin{itemize}
    \item At the end of the $k_1^{th}$ round the formed configuration is $\mathcal{B}$ from initial configuration $\mathcal{A}$ and remains till start of the $k_2^{th}$ round.
    \item At the end of the $k_2^{th}$ round the formed configuration is $\mathcal{C}$ and remains till start of the $k_3^{th}$ round.
    \item At the end of the $k_3^{th}$ round the formed configuration is $\mathcal{B}$ and remains till start of the $k_4^{th}$ round.
    \item At the end of the $k_4^{th}$ round the formed configuration is $\mathcal{A}$.
\end{itemize}

Let $r_1$ be the terminal robot in the initial configuration that is farther from the middle robot and the $r_2$ be the other terminal robot. Let consider another execution $E_1$ where from $(k_1+1)^{th}$ round to $(k_3-1)^{th}$ round the robot $r_1$ remains inactivated. Then $k_3^{th}$ round the view and state of $r_1$ for execution $E$ is the same as the view and state of $r_1$ in $k_1^{th}$ round for execution $E_1$. As according to execution $E$ of, in $E_1$ the $r_1$ will not decide to move if $r_1$ is activated in $k_3^{th}$ round. Also from the execution of $E$, $r_2$ also does not decide to move on activation in $k_3^{th}$ round. Thus, after $k_3^{th}$ round onwards the configuration shall remain the same which is a contradiction. Suppose the other way round happens. If $r_2$ moves at all without $r_1$ moving then it will create a configuration that yields a contradiction. Suppose $r_1$ moves after $m$ rounds of $k_3^{th}$ round. Then if from $k_1^{th}$ round if $r_2$ is not activated for consecutive $m$ rounds then it creates a contradictory configuration. Thus, $\mathsf{A}$ cannot solve the problem. Thus the result follows.

\paragraph{\bf Proof of Lemma~\ref{lm:fcomOSP}:}
  Following lines 1-6 of the Algorithm \textsc{AlgoOSP} the configuration transforms from $\mathcal A$ to $\mathcal B$. Following from lines 7-10 of the Algorithm \textsc{AlgoOSP} the configuration transforms from $\mathcal B$ to $\mathcal C$. Following lines 11-15 of the Algorithm \textsc{AlgoOSP} the configuration transforms from $\mathcal C$ to $\mathcal B$. Following from lines 16-18 of the Algorithm \textsc{AlgoOSP} the configuration transforms from $\mathcal B$ to $\mathcal A$. Then following lines 5-6 the configuration becomes equal to the initial configuration that is, the configuration is $\mathcal A$ with both robots color \texttt{off}. Thus, we can conclude the result.

\paragraph{\bf Proof of Lemma~\ref{lm:SUIRlumi}:}

If possible let there exist an algorithm $\mathcal A$ that solves the problem in $\mathcal{LUMI}^S$. The scheduler can activate robots in such a way that it eventually ends up with a configuration where two robots are in one hop. If in each round only one robot is activated alternatively then eventually one of the robots will move toward the other robot, otherwise, the initial configuration will remain forever. Thus eventually a configuration will be formed where two robots are one hop away from each other before gathering. After achieving such a configuration the gathering can be done only if one of the robots moves toward the other but the other does not. Suppose the colors of the robots are such that according to the algorithm say, $r_1$ is about to move. Then scheduler can choose to crash $r_1$. A similar argument can be made for the other robot. Thus, even if one robot does not move in a particular round, it has to choose to move eventually towards the other robot.

Let us end up with a configuration having robot $r_1$ with color $c_1$ and having robot $r_2$ with color $c_2$. If for such a scenario both robots are supposed to move through then the scheduler will activate both robots resulting in the same configuration at the end. Suppose for this configuration $r_1$ would move but $r_2$ cannot wait forever at that position because $r_1$ can get crashed at its position. If $r_1$ does not get activated for enough rounds then $r_2$ must move towards $r_1$ after finite rounds. Let the scheduler not activate $r_1$ for that many rounds then $r_2$ rounds. Then after that, $r_2$ after that would move towards $r_1$. In that round, if $r_1$ also chooses to move then the scheduler may choose to activate both robots which results same configuration getting formed again. If $r_1$ does not move at this round, then also $r_1$ cannot await there forever using the same argument as we used earlier. Thus in this way, the scheduler can activate the robots in such a way that the robots will not ever gather. 

\paragraph{\bf Proof of Lemma~\ref{lm:SUIRoblot}:}
We design an algorithm in which, on activation, a robot shall move towards the other robot. In the first round, if no robot crashes, then at the end of the round, the gathering will be done. Otherwise, at the end of the first round, the non-crashed robot shall move towards the crashed robot. In the next round, the non-crashed robot will again move toward the crashed robot and gather there.
%




\end{document}